\documentclass[a4paper,UKenglish]{lipics-v2016}
\usepackage{microtype}
\usepackage{amsthm}
\usepackage{amsmath}
\usepackage{amssymb}
\usepackage{mathtools}

\newcommand{\BigO}{\mathcal{O}}
\newcommand{\maj}{MAJ}

\DeclarePairedDelimiter{\ceil}{\lceil}{\rceil}
\DeclarePairedDelimiter{\floor}{\lfloor}{\rfloor}

\newtheorem{claim}{Claim}

\title{Computing majority with low-fan-in majority queries\footnote{This work was partially supported by the Russian Academic Excellence Project ’5-100’.}}

\author[1]{Gleb I. Posobin}
\affil[1]{National Research University, Higher School of Economics, Moscow, Russia\\
  \texttt{posobin@gmail.com}}
\authorrunning{G.\,I. Posobin} 

\Copyright{Gleb I. Posobin}


\EventEditors{John Q. Open and Joan R. Acces}
\EventNoEds{2}
\EventLongTitle{42nd Conference on Very Important Topics (CVIT 2016)}
\EventShortTitle{CVIT 2016}
\EventAcronym{CVIT}
\EventYear{2016}
\EventDate{December 24--27, 2016}
\EventLocation{Little Whinging, United Kingdom}
\EventLogo{}
\SeriesVolume{42}
\ArticleNo{23}

\begin{document}

\maketitle

\begin{abstract}
In this paper we examine the problem of computing majority function $\maj_n$ on $n$ bits by depth-two formula, where each gate is a majority function on at most $k$ inputs.
We present such formula that gives the first nontrivial upper bound for this problem, with $k = \frac{2}{3} n + 4$.
This answers an open question in \cite{kulikov2017computing}.

We also look at this problem in adaptive setting --- when we are allowed to query for value of $\maj_k$ on any subset, and wish to minimize the number of such queries.
We give a simple lower bound for this setting with $\ceil{n/k}$ queries, and we present two algorithms for this model: the first one makes $\approx 2\frac{n}{k} \log k$ queries in the case when we are limited to the standard majority functions, and the second one makes $\frac{n}{k} \log k$ queries when we are allowed to change the threshold of majority function.

 \end{abstract}

    \section{Introduction}
    We study the problem of computing majority function $\maj_n$ if we are only allowed to query values of $\maj_k$ (that is, majority on $k$ variables) functions.
    
    Majority functions and circuits consisting of them arise in various areas of computational complexity (see e.g. \cite{minsky, Gold01, DBLP:books/daglib/0028687}).
    Particularly, iterated (or recursive) majority that consists of iterated application of majority of small number of inputs to itself, provides an example of a function with interesting complexity properties in various models (\cite{DBLP:journals/cc/JuknaRSW99, DBLP:journals/rsa/MagniezNSSTX16, DBLP:journals/rsa/MosselO03, DBLP:journals/siamcomp/KampZ07}), and helps in various constructions.

    Another motivation for this problem comes from the studies of boolean circuits of constant depth.
    $\widehat{\textsf{TC}}{}^{0}$ is defined as the class of functions computable by constant-depth polynomial-sized circuits consisting of majority gates, and it plays one of the major roles in this area (\cite{DBLP:books/daglib/0028687}).
    
    The first model we look at is a standard boolean circuit of depth two, but where all the gates are threshold functions of at most $k$ variables.
    It was shown in \cite{DBLP:journals/eccc/EngelsGMR17} that to be able to compute such circuits $k$ must be $\Omega(n^{4/5})$.
    There is a trivial upper bound $k = n$: just take all the gates on the first level to match inputs (that is, gate $f_i = x_i = \maj_{\{i\}}(x)$), and the output gate is a standard majority $\maj_n(f_1, \ldots, f_n)$.
    \cite{kulikov2017computing} did not provide any nontrivial upper bound on $k$ and posed an open question whether such bound exists.
    We answer this question positively and present the depth-two circuit with $\maj_{\frac{2}{3}n+4}$ gates that computes $\maj_n$ function.
    This result was independently obtained by Bruno Bauwens by computer search, as told to us in personal communication.
    For the sake of completeness, we also give a simple proof of lower bound $k \geq \frac{1}{\sqrt[3]{2}}n^{2/3}$.
    
    The second model is an adaptive one --- algorithm can query the $\maj_k$ function on any subset of size at most $k$, and we are interested in the worst case number of queries made to compute $\maj_n$.
    This problem was motivated by the proof in \cite{kulikov2017computing}, which does not use any properties of the function in the output gate except for its monotonicity, and also by \cite{DBLP:conf/latin/EppsteinH16}, where they study adaptive computation of $\maj_n$ with discrepancy queries on $k$ elements.
    We give a simple proof for the lower bound of $\ceil{n/k}$ queries, and we also present two algorithms for adaptive setting: one that runs in fixed-threshold setting and which requires $2 (\frac{n}{k-4}+1)(\log k+4) \approx 2 \frac{n}{k} \log k$ queries, and another that runs in adjustable-threshold adaptive model (that is, when the algorithm can specify majority's threshold) and requires $\ceil{n/k} \ceil{\log(k+1)}$ queries.
    
    \section{Definitions}
    By $\log$ we denote binary logarithm and
    $n$ always means the length of input vector $x = (x_1, \ldots, x_n) \in \{0, 1\}^n$.
    
    We are going to study majority functions $\maj_n(x_1, \ldots, x_n) = [x_1 + \ldots + x_n \geq n/2]$, where $[P] = 1$ if condition $P$ is true, and $[P] = 0$ otherwise.
    We will also call $\maj_k$ standard majority functions.
    By $sum_S(x)$ we denote $\sum_{i\in S} x_i$.
    We define $\maj_S(x)$ for $S \subset \{1, \ldots, n\}$ as $[sum_S(x) \geq |S|/2]$. 
    We also define $\maj_S(x; t) = [sum_S(x) \geq t]$.
    
    Threshold functions are a generalization of majority functions: $f(x_1, \ldots, x_n) = [\alpha_1 x_1 + \ldots + \alpha_n x_n \geq t]$.
    
    The first (``static``) model that we look at is $\maj_k \circ \maj_k$ --- it is defined as depth-two formula, consisting of threshold gates $[\sum \alpha_i x_i \geq t]$ such that $\sum \alpha_i \leq k$ and $\alpha_i$ are positive integers.
    
    In the second (``adaptive``) model an algorithm is given access to an oracle that knows vector $x$.
    There are two variations of this model that we look at.
    
    In the adaptive model with adjustable threshold the oracle receives a pair $(S, t)$, where $S$ is a set $S \subset \{1, \ldots, n\}$ such that $|S| \leq k$ and $t$ is an integer, and responds with the value of $\maj_S(x; t)$.
    
    In the adaptive model with fixed threshold the oracle receives a set $S \subset \{1, \ldots, n\}$ such that $|S| \leq k$ and responds with the value of $\maj_S(x)$.
    
    \section{Static model}
    \subsection{Lower bound}
    \begin{lemma}
    For large enough $n$, $k$ should be at least $\frac{1}{\sqrt[3]{2}}n^{2/3}$ in order to compute $\maj_n$ in $\maj_k \circ \maj_k$ model.
    \end{lemma}
    
    This is not the best currently known lower bound, for a more involved proof of $k \geq n^{13/19 + o(1)}$ see \cite{kulikov2017computing}.
    This proof is simpler and is adapted from the proof of a stronger claim in \cite{kulikov2017computing} that shows that for some function in $\maj_k \circ \maj_k$ to differ from $\maj_n$ on small fraction (less than $1/10$) of inputs $x$ with $\sum_i x_i = \floor{(n-1)/2}$, $k$ must be $\Omega(n^{2/3})$.
    We provide this proof just for the sake of completeness.
    
    \begin{proof}
    Let's look at the boolean hypercube $H^n = \{0, 1\}^n$.
    It is known (\cite{o2014analysis}) that among all monotone boolean functions majority is the one with the largest influence --- that is, its value changes on the largest possible number of edges of the hypercube.
    The number of such edges for $\maj_k$ is $k\binom{k}{k/2} \leq k^{1/2}2^k$.
    So if some gate $f_i$ depends on $l \leq k$ variables, then it changes its value on at most $l^{1/2} 2^l$ edges of $H^l$.
    
    Obviously, on each edge $(u, v)$ in $H^n$ such that $\sum_i u_i < n/2, \sum_i v_i \geq n/2$, there must be at least one function $f_i$ in the first level of the formula, such that $f_i(u) \not = f_i(v)$.
    There are $n \binom{n}{n/2} > \frac{n^{1/2}2^n}{\sqrt{2}}$ such edges.
    
    Let $f_i$ be some gate in the first level of the formula, without loss of generality assume that $f_i$ depends only on $x_1, \ldots, x_l$ for $l \leq k$.
    Then, for each $u \in H^l$ there are at most $2^{n-l}/(n-l)^{1/2}$ possible continuations $u_{l+1}, \ldots, u_n$ of $u$ that have $\sum_{1 \leq i \leq n} u_i = \floor{(n-1)/2}$.
    This is so, since $u_{l+1}, \ldots, u_n$ must sum to $\floor{(n-1)/2} - \sum_{1 \leq i \leq l} u_i$, and the number of such assignments of $u$ is a binomial coefficient.
    The largest binomial coefficient $\binom{n-l}{z}$ for fixed $n-l$ is the central one: $z = \ceil{(n-l)/2}$, for which there is a well-known bound $\binom{n-l}{z} \leq 2^{n-l}/(n-l)^{1/2}$.
    
    This means that for a fixed $f_i$ there are at most $l \frac{2^l}{l^{1/2}} \frac{2^{n-l}}{(n-l)^{1/2}} \leq \sqrt{2} k^{1/2} \frac{2^n}{n^{1/2}}$ such edges $(u, v)$ that $\sum_i u_i < n/2, \sum_i v_i \geq n/2$ and $f_i(u) \not = f_i(v)$ (note that $(n-l) \geq \frac{1}{2}n$ for large enough $n$, if $l \leq k = \BigO(n^{2/3})$).
    Summing this for all $f_i$ we get:
    \[\sqrt{2} k^{3/2} \frac{2^n}{n^{1/2}} \geq \frac{n^{1/2}2^n}{\sqrt{2}} \implies k^3 \geq \frac{n^2}{2}\]
    \end{proof}
    
    \subsection{Upper bound}
    \begin{theorem}
    For any $n$ $\maj_n$ is computable by some circuit in $\maj_{\frac{2}{3}n+4} \circ \maj_{\frac{2}{3}n+4}$.
    \end{theorem}
    
    \begin{proof}
    
    \textbf{Case 1.} $n$ is divisible by $3$,
    $k = \frac{2}{3}n+1$.\\
    We need to define $\frac{2}{3} n + 1$ first-level threshold functions $f_i$ in the formula.
    Let $f_1 = x_1, \ldots, f_{\frac{1}{3}n} = x_{\frac{1}{3}n}$, each equivalently being a majority function on one corresponding element.
    Denote $\frac{2}{3} n$ indices of yet unused inputs by $M = \{\frac{1}{3} n + 1, \ldots, n\}$.
    Then let $f_{\frac{1}{3}n + i} = \maj_M(x; \ceil{\frac{1}{6}n} + i - 1)$.
    The output gate is a standard majority on all $f_i$.
    \[
    \left[\sum f_i \geq \frac{n}{3} + \frac{1}{2}\right] = \left[\Big(\sum_{\mathclap{i\leq n/3}}x_i + \sum_{\mathclap{0 \leq i \leq n/3}} [sum_M(x) \geq \ceil{n/6} + i]\Big) \geq \frac{n}{3} + \frac{1}{2}\right]
    \]
    If $sum_M(x) \geq \ceil{\frac{1}{6}n} + \frac{1}{3}n = \ceil{\frac{1}{2}n}$, then $\sum_{0 \leq i \leq \frac{1}{3}n} [sum_M(x) \geq \ceil{\frac{1}{6}n} + i] = \frac{1}{3}n + 1$, that means that both output gate and $\maj_n(x)$ will be equal to one.
    
    If $\ceil{\frac{1}{6}n} \leq sum_M(x) \leq \ceil{\frac{1}{2}n}$, then $\sum_{i\leq \frac{1}{3}n}x_i + \sum_{0 \leq i \leq \frac{1}{3}n} [sum_M \geq \ceil{\frac{1}{6}n} + i] = \sum_{i\leq \frac{1}{3}n}x_i + sum_M(x) - \ceil{\frac{1}{6}n} + 1 = \sum_i x_i - \ceil{\frac{1}{6}n} + 1$.
    Then the circuit attains value $1$ iff $\sum_i x_i - \ceil{\frac{1}{6}n} + 1 \geq \frac{1}{3}n + \frac{1}{2}$. Rearranging, we get that the value of the circuit is $1$ iff $\sum_i x_i \geq \frac{1}{3}n + \ceil{\frac{1}{6}n} - \frac{1}{2} = \ceil{\frac{1}{2}n} - \frac{1}{2}$, which is equivalent to $[\sum_i x_i \geq \frac{1}{2}n] = \maj_n(x)$.
    
    If $sum_M(x) < \ceil{\frac{1}{6}n}$, $\maj_n(x) = 0$ and the formula's value is also $0$, since $\sum_{i\leq \frac{1}{3}n}x_i \leq \frac{1}{3}n$.
    
    \textbf{Case 2.} $n = 3m + 1$ for some integer $m$,
    $k = 2m+3 \leq \frac{2}{3}n + 3$.\\
    Let $f_i = x_i$ for $1 \leq i \leq m+1$.
    Let $M = \{m+2, \ldots, 3m+1\}$.
    Then let $f_{m+1+i} = \maj_M(x; \ceil{\frac{1}{2}(m-1)} + i - 1)$ for $1 \leq i \leq m+2$, so there are $2m + 3$ gates in the first level of the circuit.
    The output gate is again a majority on all $f_i$.
    
    If $sum_M(x) \geq \ceil{\frac{1}{2}(m-1)} + m + 1 = \ceil{\frac{3}{2}m+\frac{1}{2}}$, then both $\maj_n(x)$ and the output gate are equal to one.
    
    If $\ceil{\frac{1}{2}(m-1)} \leq sum_M(x) \leq \ceil{\frac{1}{2}(m-1)} + m + 1$, then $\sum_i x_i = \sum_{i\leq m+1}x_i + sum_M(x) - \ceil{\frac{1}{2}(m-1)} + 1 = \sum_i x_i - \ceil{\frac{1}{2}(m-1)} + 1$.
    Then the circuit attains value $1$ iff $\sum_i x_i - \ceil{\frac{1}{2}(m-1)} + 1 \geq m + \frac{3}{2}$.
    Rearranging, we get that the value of the circuit is $1$ iff $\sum_i x_i \geq m + \ceil{\frac{1}{2}(m-1)} + \frac{1}{2} = \ceil{\frac{3}{2}m + \frac{1}{2}} - \frac{1}{2}$.
    Rounding the last part to the smallest greater integer we get that the circuit attains value $1$ iff $\sum_i x_i \geq \ceil{\frac{3}{2}m + \frac{1}{2}}$, exactly when $\maj_n(x) = 1$.
    
    If $sum_M(x) < \ceil{\frac{1}{2}(m-1)}$, then $\sum_i x_i < m + 1 + \ceil{\frac{1}{2}(m-1)} = \ceil{\frac{3}{2}m+1/2}$, so both $\maj_n(x)$ and the circuit will have value $0$.
    
    \textbf{Case 3.} $n = 3m + 2$ for some integer $m$,
    $k = 2m+5 \leq \frac{2}{3}n + 4$.\\
    Let $f_i = x_i$ for $1 \leq i \leq m+2$.
    Let $M = \{m+3, \ldots, 3m+2\}$.
    Then let $f_{m+2+i} = \maj_M(x; \ceil{\frac{1}{2}m-1} + i - 1)$ for $1 \leq i \leq m+3$, so there are $2m + 5$ gates in the first level of the circuit.
    The output gate is again a majority on all $f_i$.
    
    If $sum_M(x) \geq \ceil{\frac{1}{2}m-1} + m + 2 = \ceil{\frac{3}{2}m}+1 = \ceil{\frac{3}{2}m+1}$, then both $\maj_n(x)$ and the output gate are equal to one.
    
    If $\ceil{\frac{1}{2}m-1} \leq sum_M(x) \leq \ceil{\frac{1}{2}m-1} + m + 2 = \ceil{\frac{3}{2}m+1}$, then $\sum_i x_i = \sum_{i\leq m+2}x_i + sum_M(x) - \ceil{\frac{1}{2}m-1} + 1 = \sum_i x_i - \ceil{\frac{1}{2}m} + 2$.
    Then the circuit attains value $1$ iff $\sum_i x_i - \ceil{\frac{1}{2}m} + 2 \geq m + \frac{5}{2}$.
    Rearranging, we get that the value of the circuit is $1$ iff $\sum_i x_i \geq m + \ceil{\frac{1}{2}m} + \frac{1}{2} = \ceil{\frac{3}{2}m + 1} - \frac{1}{2}$.
    Rounding the last part to the smallest greater integer we get that the circuit attains value $1$ iff $\sum_i x_i \geq \ceil{\frac{3}{2}m + 1}$, exactly when $\maj_n(x) = 1$.
    
    If $sum_M(x) < \ceil{\frac{1}{2}m-1}$, then $\sum_i x_i < m + 2 + \ceil{\frac{1}{2}m-1} = \ceil{\frac{3}{2}m+1}$, so both $\maj_n(x)$ and the circuit will have value $0$.

    \end{proof}
    
    \section{Adaptive model}
    \subsection{Lower bound}
    \begin{lemma}
    Any algorithm for computing $\maj_n$ in adaptive model requires $\ceil{n/k}$ queries to oracle.
    \end{lemma}
    \begin{proof}
    We will set bits of vector $x$ the algorithm has not yet asked about right before answering a query from the algorithm.
    Initialize a set $S = \varnothing$ --- this will be the set of all the indices that appeared in algorithm's requests.
    When algorithm makes a request for $\maj_{T}(x)$, we firstly set $x_i$ for $i$ in $T \backslash S$: set any $z = \floor{|T \cup S| / 2} - sum_S(x)$ variables $x_i$ to $1$ and the rest $|T \backslash S| - z$ variables to $0$.
    Then we update $S := S \cup T$.
    After that we answer the query with $\maj_T(x)$.
    
    Clearly, $sum_S(x) = \floor{|S|/2}$ at each step.
    If algorithm has made less than $\ceil{n/k}$ requests, then $|S| < n$, so there are some variables
    that the algorithm knows nothing about.
    
    There are $n - |S|$ such variables.
    If we set $x_i$ with $i \not \in S$ to $1$, then $\sum_i x_i = \floor{|S|/2} + n - |S| = n - \ceil{|S|/2} \geq n - \ceil{(n-1)/2} \geq n - n/2 = n/2$, and $\maj_n(x)$ will be equal to $1$.
    If we set $x_i$ with $i \not \in S$ to $0$, then $\sum_i x_i = \floor{|S|/2} \leq \floor{(n-1)/2} \leq (n-1)/2 < n/2$, and $\maj_n(x)$ will be equal to $0$.
    
    So, whichever answer the algorithm chooses, there will be two vectors $x$ and $y$ consistent with all our answers, such that $\maj_n(x) \not = \maj_n(y)$, so on either $x$ or $y$ the algorithm will err.
    \end{proof}
    
    \subsection{Upper bound in adjustable-threshold setting}
    \begin{lemma}
    There exists an algorithm for adjustable-threshold adaptive model that determines $\maj_n(x)$ in $\ceil{\frac{n}{k}}\ceil{\log(k+1)}$ queries to the oracle.
    \end{lemma}
    \begin{proof}
    Split $\{1, \ldots, n\}$ into $\ceil{\frac{n}{k}}$ disjoint blocks $B_i$ of size at most $k$ each.
    We know that $\maj_{B_i}(x; 0) = 1, \maj_{B_i}(x; |B_i|+1) = 0$, so using binary search we can find such value $h_i$ that $\maj_{B_i}(x; h_i) = 1, \maj_{B_i}(x; h_i+1) = 0$.
    This means that $sum_{B_i}(x) = h_i$.
    Since binary search works in $\ceil{\log(|B_i|+1)}$ steps, in $\ceil{\frac{n}{k}}\ceil{\log(|B_i|+1)}$ queries we will know $\sum_i x_i = \sum_i sum_{B_i}(x) = \sum_i h_i$, and so we will easily find $\maj_n(x) = [\sum_i h_i \geq \frac{n}{2}]$.
    \end{proof}
    
    The upper bound for the adaptive model with adjustable threshold turns out to be at most two times better than the one we get in the case of the adaptive model with fixed threshold.
    
    \subsection{Upper bound in fixed-threshold setting}
    \begin{lemma} \label{binsearch}
    There exists an algorithm such that it accepts two disjoint sets $A, B$, $|A| = |B| \leq k$, $\maj_A(x) \not = \maj_B(x)$, and if given access to the oracle from the definition of fixed-threshold adaptive model it finds in at most $\ceil{\log(|A|+1)}$ queries to the oracle a set $S \subset A \cup B$ such that $sum_S(x) = |S|/2$ and $|S| \geq |A|-1$.
    \end{lemma}
    
    \begin{proof}
    We will define such procedure \texttt{find\_balanced\_set(A, B, x)}.
    
    Let $A_1, \ldots, A_{|A|}$ be some enumeration of elements from $A$ and $B_1, \ldots, B_{|B|}$ be some enumeration of elements from $B$.
    
    Define $c=|A|$, $i(m) = A_m$ if $m \leq c$, otherwise $i(m) = B_{m-c}$, and $v(m) = x_{i(m)}$.
    Also define $f(m) = \maj_c(v(m), \ldots, v(m+c-1))$.
    From the specification of \texttt{find\_balanced\_set(A, B, x)} we know that $f(1) \not = f(c+1)$, so we can use binary search to find position $h$ such that $f(1) = f(h) \not = f(h+1) = f(c+1)$.
    This implies that $v(h) \not = v(h+c)$, and, moreover, $v(h) = f(h)$.
    
    If $c$ is even, this implies $sum_{\{i(h+1), \ldots, i(h+c-1)\}} = \frac{c}{2}-1$, so either $sum_{\{i(h), \ldots, i(h+c-1)\}} = \frac{c}{2}$, if $f(h) = 1$, or $sum_{\{i(h+1), \ldots, i(h+c)\}} = \frac{c}{2}$, if $f(h) = 0$.
    The procedure \texttt{find\_balanced\_set(A, B, x)} then returns the set $\{i(h+1-f(h)), \ldots, i(h+c-f(h))\}$.
    
    If $c$ is odd, $sum_{\{i(h+1), \ldots, i(h+c-1)\}} = \frac{c-1}{2}$, so the procedure returns the set $\{i(h+1), \ldots, i(h+c-1)\}$.
    
    Clearly, this procedure makes at most $\ceil{\log(c+1)}$ queries.
    \end{proof}
    
    \begin{claim} \label{balanced_claim}
    Suppose that for some set $S \subset \{1, \ldots, n\}$ we know that $sum_S(x) = |S|/2$.
    Then, $\maj_n(x) = \maj_{\{1, \ldots, n\} \backslash S}(x)$, so we can forget about indices from $S$ and not take them into consideration, as though vector $x$ has size $n - |S|$.
    \end{claim}
    
    The basic idea of our adaptive algorithm is simple: on each iteration we find some large set $S$ with $sum_S(x) = |S|/2$, and remove the variables with indices in $S$ by claim \ref{balanced_claim}.
    We find such set $S$ by splitting $\{1, \ldots, n\}$ into $n/k$ subsets of similar sizes, querying each of them, finding two sets of sizes $\approx |S|$ with different $\maj$ values on them, and running the \texttt{find\_balanced\_set} procedure from lemma \ref{binsearch}.
    This is the main idea, but we have to be careful with the sizes of subsets, since $k$ may not divide $n$ evenly, and if we are too careless, after several iterations we may have to split the set of current indices again and make queries for each of new subsets, which could significatnly increase the total number of queries.
    
    \begin{theorem}
    There exists an algorithm for computing $\maj_n$ function in the fixed-threshold adaptive setting using at most $2 (\frac{n}{k-4}+1)(\log k+4)$ queries.
    \end{theorem}
    
    \begin{proof}
    Split the set $\{1, \ldots, n\}$ into $\ceil{\frac{n}{k}}$ disjoint sets $S_i$ of sizes $\ell = \ceil{\frac{n}{\ceil{n/k}}}$ or $\ell - 1$ each.
    Enumerate $S_i$ in such way that $S_1, \ldots, S_a$ are of size $\ell$ each, and the rest are of size $\ell - 1$.
    $\ell \leq k$, since $\frac{n}{\ceil{n/k}} \leq \frac{n}{n/k} = k$.
    For all $S_i$ query $\maj_{S_i}(x)$.
    If all the answers are the same and equal to $y \in \{0, 1\}$, then $\maj_n(x) = y$ and we are done.
    
    \textbf{Case 1.} \label{case1} Suppose that there are sets $S_i$ of both sizes $\ell$ and $\ell - 1$.
    Then, since not all the answers are the same, there are two indices $i$ and $j$ such that $|S_i| = \ell$, $|S_j| = \ell - 1$ and $\maj_{S_i}(x) \not = \maj_{S_j}(x)$ (otherwise, all $S_i$ with size $\ell$ have $\maj_{S_i}(x) = \maj_{S_j}(x)$ for all $S_j$ of size $\ell - 1$, so all the subsets $S_i, S_j$ have the same answer, a contradiction).
    
    Now we will find a large (with size $\approx \ell$) subset $S \subset \{1, \ldots, n\}$ of indices to exclude from our consideration (by claim \ref{balanced_claim}) in at most logarithmic in $\ell$ number of queries.
    
    Choose any element $y$ from $S_i$ and query $\maj_{S_i\backslash\{y\}}(x)$.
    Suppose that $\maj_{S_i}(x) \not = \maj_{S_i\backslash \{y\}}(x)$.
    Clearly, $x_y$ then is equal to $\maj_{S_i}(x)$.
    
    If $\maj_{S_i}(x)$ is $0$, then $sum_{S_i\backslash \{y\}}(x) \geq \frac{\ell - 1}{2}$ and $sum_{S_i\backslash \{y\}}(x) < \frac{\ell}{2}$, implying $\ell$ is odd and $sum_{S_i\backslash \{y\}}(x) = \frac{\ell - 1}{2}$.
    This means that we can ignore indices from $S_i\backslash \{y\}$ by claim \ref{balanced_claim}.
    
    If $\maj_{S_i}(x)$ is $1$, then $sum_{S_i\backslash \{y\}}(x) < \frac{\ell - 1}{2}$ and $sum_{S_i\backslash \{y\}}(x) + 1 \geq \frac{\ell}{2}$, implying $\ell$ is even and $sum_{S_i}(x) = \frac{\ell}{2}$.
    Again, this means that we can remove indices from $S_i$ from our consideration by claim \ref{balanced_claim}.
    
    So, if there are sets $S_i$ with different sizes, we are left with the case when $\maj_{S_i}(x) = \maj_{S_i\backslash \{y\}}(x)$.
    
    Let $M = S_i\backslash \{y\}$.
    We know that $|M| = |S_j| = \ell - 1$ and that $\maj_M(x) \not = \maj_{S_j}(x)$.
    
    We run \texttt{find\_balanced\_set} from lemma \ref{binsearch} on sets $M$ and $S_j$, and get the set $S$ such that $sum_S(x) = |S|/2$. Note that $S$ will be of size $\ell-2$ (in case when $\ell$ is even) or $\ell-1$.
    Since total number of elements in $M \cup S_j$ is $2\ell - 2$, $|M \cup S_j \backslash S| \in \{\ell-1, \ell\}$.
    
    We then forget about indices in $S$ by claim \ref{balanced_claim}, query $\maj_{M \cup S_j \backslash S}(x)$ and repeat the \hyperref[case1]{\textbf{case 1}} until all $S_i$ are of the same size or have the same answer.
    
    After several iterations we either get sets $S_i$ with the same answer $\maj_{S_i}(x)$ each and we are done, or we get the sets with the same sizes $|S_i|$.
    
    \textbf{Case 2.} All of the sets $S_i$ have same sizes.\\
    Take two sets $S_i$ and $S_j$ with different answers $\maj_{S_i} \not = \maj_{S_j}$ and run the procedure \texttt{find\_balanced\_set(S\_i, S\_j, x)}.
    Let $S$ be the answer.
    If $|S_i|$ is even, $|S| = |S_i|$, so just exclude $S$ from further considerations, query the oracle $\maj_{S_i \cup S_j \backslash S}(x)$, and repeat the procedure until all the sets left have the same $\maj_{S_i}(x)$.
    
    The case of odd $|S_i|$ requires more careful consideration, since if we simply leave elements from $S$ (which has size $|S| = |S_i|-1$) out, $|S_i \cup S_j \backslash S|$ will be $|S_i|+1$, and that may become greater than $k$ after several iterations.
    If we look at the case of odd $c$ in \texttt{find\_balanced\_set}, we will see that instead of the set $\{i(h+1), \ldots, i(h+c-1)\}$ of size $c-1$ it can as well return the set $\{i(h), \ldots, i(h+c)\}$ of size $c+1$, since $v(h) \not = v(h+c)$.
    
    So if $|S_i \cup S_j \backslash S| > k$, then we add elements $i(h), i(h+c)$ to $S$: $S := S \cup \{i(h), i(h+c)\}$.
    We can remove indices from $S$ from future considerations by claim \ref{balanced_claim}, and we query the oracle the value of $\maj_{S_i \cup S_j \backslash S}(x)$.
    Now we again have sets $S_i$ of two different sizes, so we return to \hyperref[case1]{\textbf{case 1}}.
    
    The analysis of number of queries this algorithm makes is rather straightforward:
    it starts with making $\ceil{\frac{n}{\ell-1}}$ queries, then on each iteration it discards at least $\ell-2$ indices after making at most $\ceil{log(\ell+1)}+2$ queries.
    So the total number of queries made is $\ceil{\frac{n}{\ell-1}} + \ceil{\frac{n}{\ell-2}}(\ceil{\log(\ell+1)}+2)$.
    As was stated in the beginning, $\ell \leq k$.
    On the other hand, $\ell \geq \frac{n}{\ceil{n/k}} \geq \frac{n}{1+n/k} = \frac{1}{1+k/n} k \geq \frac{k}{2}$.
    This gives the following upper bound on the number of queries:
    \[2 \frac{n}{k-2} + 1 + 2 (\frac{n}{k-4}+1)(\log k+3) \leq 2 (\frac{n}{k-4}+1)(\log k+4)\]
    \end{proof}
    
    \section{Conclusion}
    We have presented new upper bounds for computing majority functions $\maj_n(x)$ in two models.
    
    In static model we have constructed the depth-two circuit for $\maj_n$ with $\maj_{\frac{2}{3}n+4}$ gates, which gives the first nontrivial upper bound for this problem.
    
    We have defined an adaptive model with two variations, and presented an algorithm for computing $\maj_n$ function using $2 (\frac{n}{k-4}+1)(\log k+3) \approx 2 \frac{n}{k} \log k$ queries to the oracle in the fixed-threshold setting, and a simple algorithm for $\maj_n$ with $\ceil{n/k}\ceil{\log(k+1)} \approx \frac{n}{k} \log k$ queries in the adjustable-threshold setting.
    
    There is a large gap between current lower and upper bounds for the static model: the best lower bound \cite{DBLP:journals/eccc/EngelsGMR17} is $k \geq \Omega(n^{4/5})$, and our upper bound is $\frac{2}{3}n + 4$.
    In the adaptive model there is also a gap between the bounds, albeit a logarithmic (in $k$) one: our lower bound is $\ceil{n/k}$.
    So the natural questions for future consideration are those of narrowing these gaps.
    
    Another interesting possible direction in the adaptive case is the addition of noise to oracle's answers --- with some probability $\varepsilon$ oracle gives us a uniformly random answer.
    
    \section{Acknowledgements}
    I would like to thank Vladimir Podolskii for posing the problem and for fruitful discussions.

\bibliography{bibliography}


\end{document}